\documentclass[a4paper,11pt,leqno]{amsart}
\usepackage{amsmath}
\usepackage{amsfonts}
\usepackage{eucal}
\usepackage{amsthm}

\newcommand{\bigpare}[1]{\bigl(#1\bigr)}

\newcommand{\Bigbra}[1]{\Bigl\{#1\Bigr\}}

\newcommand{\bigbrac}[1]{\bigl[#1\bigr]}

\newcommand{\norm}[1]{\| #1 \|}
\newcommand{\bignorm}[1]{\bigl\| #1 \bigr\|}

\newcommand{\jap}[1]{\langle #1 \rangle}

\def\a{\alpha}
\def\b{\beta}
\def\c{\gamma}
\def\d{\delta}
\def\e{\varepsilon}

\def\i{\mbox{\raisebox{.5ex}{$\chi$}}}

\def\l{\lambda}
\def\m{\mu}

\def\s{\sigma}

\def\x{\xi}
\def\y{\eta}

\def\th{\theta}

\def\re{\mathbb{R}}
\def\co{\mathbb{C}}
\def\ze{\mathbb{Z}}

\def\pa{\partial}

\renewcommand{\Re}{\text{{\rm Re}\;}}
\renewcommand{\Im}{\text{{\rm Im}\;}}

\newcommand{\trace}{\text{{\rm Tr}} }

\newcommand{\Slash}[1]{{\ooalign{\hfil#1\hfil\crcr\raise.167ex\hbox{/}}}}

\newtheorem{thm}{Theorem}
\newtheorem{lem}[thm]{Lemma}
\newtheorem{prop}[thm]{Proposition}
\newtheorem{cor}[thm]{Corollary}

\theoremstyle{definition}

\theoremstyle{remark}
\newtheorem{rem}{Remark\!}

\begin{document}

\title[The Spectral Shift Function and The Friedel Sum Rule]{The Spectral Shift Function and \\ 
The Friedel Sum Rule}
\author[M. Kohmoto]{Mahito Kohmoto}
\author[T. Koma]{Tohru Koma}
\author[S. Nakamura]{Shu Nakamura}
\address[Kohmoto]{The Institute for Solid State Physics, The University of Tokyo, 
5-1-5 Kashiwanoha, Kashiwa, Chiba 277-8581, JAPAN}
\address[Koma]{Department of Physics, Gakushuin University, 
Mejiro, Toshima-ku, Tokyo 171-8588, JAPAN}
\email{tohru.koma@gakushuin.ac.jp}
\address[Nakamura]{Graduate School of Mathematical Sciences, 
University of Tokyo, 3-8-1 Komaba, Meguro Tokyo, 153-8914, JAPAN}
\email{shu@ms.u-tokyo.ac.jp}
\date{\today}
\maketitle 

\begin{abstract}
We study the relationship between the spectral shift function 
and the excess charge in potential scattering theory. 
Although these quantities are closely related to each other, 
they have been often formulated in different settings so far. 
Here we first give an alternative construction of  the spectral shift function, 
and then we prove that the spectral shift function thus constructed yields 
the Friedel sum rule.
\end{abstract}

\section{Introduction}

In physical systems, universal nature often reflects 
the global, geometrical (topological) structure of the system. 
For example, Gauss's law in classical electromagnetism is a consequence of 
the geometrical structure of the three-dimensional Euclidean space. 
It states that the flux $\Phi$ of the electric field ${\bf E}$ through 
any closed surface $\Sigma$ is proportional to the total charge $Q$  
enclosed by the surface $\Sigma$: 
\[
\Phi=\int_\Sigma{\bf E}\cdot d{\bf a}=C Q
\]
with a constant $C$. 

In this paper, we study an analogue to Gauss's law in 
scattering theory in quantum mechanics. 
Let us consider a metal with a single impurity at zero temperature. 
The impurity potential scatters the conduction electrons, 
and changes their charge distribution. 
For a fixed Fermi energy ${\mathcal E}_F$, 
the ``excess charge" $Z({\mathcal E}_F)$ due to the impurity 
is defined to be the difference 
between the total numbers of levels in the Fermi sea  
with and without the impurity.  
Then the excess charge $Z({\mathcal E}_F)$ equals the total phase shifts $\theta({\mathcal E}_F)$ of 
the scattering matrix $S({\mathcal E}_F)$ for the impurity potential:
\begin{equation}
\theta({\mathcal E}_F):=\frac{1}{2\pi i}\log {\rm det}S({\mathcal E}_F)=Z({\mathcal E}_F).
\label{FSR}
\end{equation}
This is known as the Friedel sum rule \cite{Friedel} in solid state physics \cite{Kittel}. 

Since the excess charge $Z({\mathcal E}_F)$ is formally written in terms of 
the trace of the difference between the spectral projection operators with and 
without the impurity potential, it is closely related to the spectral shift function 
(SSF) which was initiated by Lifshitz \cite{Lifshitz}, and then rigorously defined by Krein \cite{Kr}. 
We briefly describe the previous construction of the SSF. 
Let $H$ and $H_0$ be a pair of self-adjoint operators. Then the SSF $\x(\cdot)$ 
is defined as a function on $\re$ satisfying the following property: 
If $f\in C_0^\infty(\re)$, then
\[
\trace\left[f(H)-f(H_0)\right]=-\int f'(\l) \x(\l) d\l.
\]
Here we note that this formula fixes $\x(\cdot)$ up to an additive constant. 
The SSF is known to exist\footnote{\ See, e.g., Birman-Yafaev \cite{BY} or Yafaev \cite{Y,Y2}.} 
if, for example,  $(H+i)^{-m}-(H_0+i)^{-m}$ is a trace class operator with some $m>0$.

Formally, the SSF is written 
\begin{equation}
\x(\l) =\trace\left[E_H(\l)-E_{H_0}(\l)\right],
\label{xiformal}
\end{equation}
where $E_A(\cdot)$ denotes the spectral projection of a self-adjoint
operator $A$. (This formal expression (\ref{xiformal}) is nothing but the excess charge!)  
It is well-known, however, $E_H(\l)-E_{H_0}(\l)$ is not necessarily in
the trace class, even when the above assumption is satisfied \cite{KoM,Kr}. 

As is well known, there are two standard constructions of the SSF.\footnote{\ See also Pushnitski~\cite{P} 
and references therein for a more sophisticated representation of the SSF.} 
The first one is due to Krein who defines the SSF 
as a locally $L^1$ function on $\re$. This construction requires relatively weak assumptions, 
and the definition is global in $\l$.
However, the existence of $\x(\l)$ for a fixed $\l\in\re$ is not obvious in this construction.
The other construction is to compute the difference of the spectral functions.
Namely, under certain conditions, one can define
\[
\x'(\l) = \trace\left[E_H'(\l)- E_{H_0}'(\l)\right]
\]
for $\l$ in a ``regular'' energy region. This method is widely used in the semiclassical
and microlocal study of the SSF.\footnote{\ See Robert \cite{R} and references therein.}
The advantage of this method is that one can study the behavior of
$\x(\l)$ in detail locally in $\l$.  On the other hand, $\x(\l)$ is not defined
globally in $\l$,
and the method requires slightly stronger assumptions on the perturbation. 

We also remark that the behavior of finite-volume spectral shift functions 
for a large volume is studied in refs.~\cite{GN1,GN2,HM,HM2,Kirsch,R,RoS}. 
In particular, under a certain condition, a sequence of finite-volume spectral 
shift functions is shown to converge to the SSF 
in the infinite-volume limit \cite{GN1,GN2,HM,HM2}. 

We propose another construction of the spectral shift function, 
$\x(\l)=\x(\l;H,H_0)$,  
for a pair of Hamiltonians, $H=-\triangle+V$ and $H_0=-\triangle$, 
on $L^2(\re^n)$. We assume that the potential $V$ satisfies
\begin{equation}
|V(x)|\leq C\jap{x}^{-\a}, \quad x\in\re^n
\label{assumpV}
\end{equation}
with some $\a>n+3$ and some $C>0$, where we have written $\jap{x}:=\sqrt{1+|x|^2}$. 
The idea for our construction is to show the existence of the boundary value 
of the perturbation determinant directly using the stationary scattering theory. 
This is a variation of Krein's construction, but we can prove that 
$\x(\l)$ is defined for each $\l\in(0,\infty)$ and 
continuous in the same region.

As an application, we consider the Friedel sum rule. We first define 
the finite-volume excess charge $Z_R(\lambda)$ due to the impurity potential $V$ by 
\[
Z_R(\lambda):={\rm Tr}\left[\vartheta_R(E_H(\lambda)
-E_{H_0}(\lambda))\vartheta_R\right], 
\]
where $\vartheta_R(x)=\vartheta(x/R)$ is a cutoff function 
with a large $R$ and with $\vartheta\in C_0^\infty(\re^n)$ 
satisfying $\vartheta=1$ in a neighborhood of $x=0$. Then we can prove 
\[
Z(\lambda):=\lim_{R\rightarrow\infty}Z_R(\lambda)
=\x(\lambda)\quad \mbox{for }\ \lambda\in(0,\infty).
\]
Namely, the excess charge $Z(\cdot)$ in the infinite-volume limit is equal to 
the SSF $\x(\cdot)$. 
On the other hand, the total phase shift $\theta(\cdot)$ for 
the scattering matrix $S(\cdot)$ is equal to $\x(\cdot)$ from the Birman-Krein formula. 
{From} these, we obtain that the Friedel sum rule (\ref{FSR}) holds 
for ${\mathcal E}_F\in(0,\infty) $ in arbitrary dimensions. 

The present paper is organized as follows: 
In Section~\ref{sec:SSF}, we first describe our method to construct 
the SSF in three or lower dimensions, and then extend it to higher dimensions.
In Section~\ref{sec:FSR}, we prove that the SSF is equal to the excess charge.

\section{Construction of the Spectral Shift Function}
\label{sec:SSF}

We construct the SSF for potential scattering theory. 
First, we describe our abstract scheme for the construction, and then prove 
the existence of the SSF. Consider a pair of Hamiltonians, 
\[
H=H_0+V, \quad H_0=-\triangle \quad \text{on} \ L^2(\re^n).
\]
We suppose the potential $V$ satisfies the bound (\ref{assumpV}) with $\alpha>n+3$. 
We may allow $V$ to have some singularities, but assume that 
it is bounded for simplicity. By the invariance principle, we construct $\x(\l)$ as
\[
\x(\l;H,H_0) = -\x((\l+M)^{-\ell}; (H+M)^{-\ell}, (H_0+M)^{-\ell})
\]
with some integer $\ell>0$ and a sufficiently large $M>0$, where
$\x(\l;A,A_0)$ denotes the SSF for a pair $A$ and $A_0$.  
Here we choose $M$ so that both $A$ and $A_0$ are bounded. We recall the
SSF is defined as
\[
\x(\l;A,A_0) = - \lim_{z\to \l+i0} \frac{1}{\pi} \Im \log \Delta_{A/A_0}(z),
\]
where $\Delta_{A/A_0}(z)$ denotes the perturbation determinant defined by
\[
\Delta_{A/A_0}(z) = \det\left[(A-z)(A_0-z)^{-1}\right] \quad \text{for} \
z\in\co\setminus(\s(A)\cup\s(A_0)).
\]
It is easy to see that $\Delta_{A/A_0}(z)$ is well-defined if $A-A_0$ is
trace class, and that it is analytic in $z$. Moreover,
\[
\Delta_{A/A_0}^{-1}(z)\Delta_{A/A_0}'(z) =-\trace\left[(A-z)^{-1}-(A_0-z)^{-1}\right],
\]
and $\x(\l;A,A_0) =0$ if $\l>\sup(\s(A)\cup\s(A_0))$ 
[or if $\l<\inf(\s(A)\cup\s(A_0))$]. Hence we have an expression of the SSF:
\begin{equation}
\x(\l;A,A_0) = \lim_{z\to\l+i0} \frac{1}{\pi}\> \Im \int_{\c_z}
\trace\left[(A-w)^{-1}-(A_0-w)^{-1}\right]dw,
\label{expSSF}
\end{equation}
where $\c_z$ denotes a contour in $\co_+:=\{z\in\co \; |\; \Im z>0\}$ such that
$\c_z(0)=k>\sup(\s(A)\cup\s(A_0))$ and $\c_z(1)=z$.
Note that this expression is consistent with the formal formula (\ref{xiformal}) by
virtue of Stone formula.

\subsection{Dimensions $n\leq 3$}

First, we prove the existence of $\x(\cdot)$ in dimensions $n\leq 3$. 
In the next section, we treat the case in dimensions $n\ge 4$. 
In this section, we set $\ell=1$, namely,
\begin{equation}
A=(H+M)^{-1}, \quad A_0=(H_0+M)^{-1}
\label{AA0}
\end{equation}
with a sufficiently large (fixed) $M>0$ so that both $A$ and $A_0$ are bounded. 
Then it is well-known that 
$A-A_0\in\mathcal {I}_1$,
where $\mathcal {I}_p$ denotes the $p$-th trace ideal.\footnote{\label{traceideal}\ See, e.g., \cite{RS}, 
Vol.3, Appendix~2 to Section~XI.3 for the criterion for the trace ideal.} 
Hence, $\Delta_{A/A_0}(z)$ is well-defined, and the above definition applies.
Now the key estimate of our construction is the following: We denote
\[
\m(z)= (z+M)^{-\ell} = (z+M)^{-1}.
\]

\begin{prop}
\label{basicpro}
Let $\l\in(0,\infty)$. Then
\[
\lim_{z\to \l+i0} \trace\left[(A-\m(z))^{-1} - (A_0-\m(z))^{-1}\right]
\]
exists, and the limit is continuous in $\l$ in $(0,\infty)$.
\end{prop}

\begin{rem}
We do not prove (or claim)
$(A-\m(\l+i0))^{-1}-(A_0-\m(\l+i0))^{-1}\in\mathcal {I}_1$.
We only prove the existence of the limit of the trace.
\end{rem}

Now combining Proposition~\ref{basicpro} with the formula (\ref{expSSF}), 
we obtain an alternative proof of the following result on the SSF:\footnote{\ For 
the cases $n=2,3$, see, e.g., \cite{Y2}, Theorem~9.1.14.}  

\begin{cor}
The SSF $\x(\l)$ exists for $\l\in (0,\infty) $, and $\x(\cdot)$
is continuous in $(0,\infty) $.
\end{cor}

Throughout the present paper, we fix $\b$ so that 
\begin{equation}
 3/2<\b<(\a-n)/2,
\label{conbeta} 
\end{equation}
and we define 
\begin{equation}
W:= \jap{x}^\b (A-A_0)\jap{x}^\b.
\label{defW}
\end{equation}

\begin{proof}[Proof of Proposition~\ref{basicpro}]
In the present case, we have 
\[
W=-\jap{x}^\b (H+M)^{-1} V(H_0+M)^{-1}\jap{x}^\b
\]
from (\ref{AA0}). Therefore, from the assumption (\ref{assumpV}) on the potential $V$ 
and the above condition (\ref{conbeta}) for $\beta$, 
we get $W\in \mathcal{I}_1$ by using the standard commutator computations. 

On the other hand, for $z\notin \s(H)\cup \s(H_0)$, we have 
\[
(A-\m(z))^{-1}-(A_0-\m(z))^{-1} =-(A-\m(z))^{-1}(A-A_0)(A_0-\m(z))^{-1}
\in\mathcal {I}_1.
\]
Combining this, the definition (\ref{defW}) of $W$ and the above result $W\in \mathcal{I}_1$, 
we obtain
\begin{align}
\label{diffResolA}
&\trace\left[(A-\m(z))^{-1}-(A_0-\m(z))^{-1}\right] \\ \nonumber
&\quad = -\trace\left[(A-\m(z))^{-1} \jap{x}^{-\b} W \jap{x}^{-\b}
(A_0-\m(z))^{-1}\right]\\ \nonumber
&\quad = -\trace\left[W
\jap{x}^{-\b}(A_0-\m(z))^{-1}(A-\m(z))^{-1}\jap{x}^{-\b}\right].
\end{align}
Now in order to complete the proof, it suffices to show the following lemma.
\end{proof}

\begin{lem}
\label{lem:2RB}
For $\l\in(0,\infty)$,
\[
\lim_{z\to\l+i0}\jap{x}^{-\b} (A_0-\m(z))^{-1}(A-\m(z))^{-1}\jap{x}^{-\b}
\]
exists in $B(L^2(\re^n))$, and the limit is continuous in $\l$ in 
$(0,\infty) $.
\end{lem}

\begin{proof}
At first we note that there are no positive eigenvalues\footnote{\ See \cite{Kato} or  \cite{RS}, Vol.4,  
Theorem~XIII.58.} under our assumption. Hence $\l$ is not an eigenvalue. We have 
\[
A_0-\m(z)=(H_0+M)^{-1} -(z+M)^{-1} =-(z+M)^{-1}(H_0-z)(H_0+M)^{-1}
\]
and hence
\begin{align*}
(A_0-\m(z))^{-1} &= -(z+M)(H_0+M)(H_0-z)^{-1} \\
&= -(z+M)-(z+M)^2(H_0-z)^{-1}.
\end{align*}
Similarly, we have
\begin{align*}
(A-\m(z))^{-1} &= -(z+M)-(z+M)^2(H-z)^{-1} \\
&=- (z+M) -(z+M)^2(H_0-z)^{-1} \\
&\qquad +(z+M)^2 (H_0-z)^{-1}V(H-z)^{-1}.
\end{align*}
Thus we have
\begin{multline*}
(A_0-\m(z))^{-1} (A-\m(z))^{-1} 
= a_0(z) +a_1(z)(H_0-z)^{-1} +a_2(z) (H_0-z)^{-2}\\ 
+ a_3(z) (H_0-z)^{-1} V(H-z)^{-1}+a_4(z) (H_0-z)^{-2} V(H-z)^{-1},
\end{multline*}
where $a_j(z)$ are polynomials in $z$. Recall $\beta>3/2$ in the condition (\ref{conbeta}) 
for $\beta$. Since
\[
\jap{x}^{-\c}(H_0-z)^{-1}\jap{x}^{-\c}, \
\jap{x}^{-\b}(H_0-z)^{-2}\jap{x}^{-\b},\ \text{and
}\jap{x}^{-\c}(H-z)^{-1}\jap{x}^{-\c}
\]
(with $\c>1/2$) are bounded and continuous\footnote{\ See Agmon~\cite{A} or Reed-Simon~\cite{RS} 
Vol.4, Section~XIII.8.} in a complex neighborhood of $\l$ in $\co_+$, we conclude the assertion.
\end{proof}

\subsection{Dimensions $n\geq 4$}

If $n\geq 4$, we set
\[
A=(H+M)^{-\ell}, \quad A_0=(H_0+M)^{-\ell}
\]
with $\ell\in\ze$ such that $n/2-1< \ell\leq n/2$. Then we have
\begin{align*}
A-A_0 &= -\sum_{j=1}^\ell (H+M)^{-j} V (H_0+M)^{-\ell-1+j} \\
&= -\sum_{j=1}^\ell (H_0+M)^{-j} V (H_0+M)^{-\ell-1+j} \\
&\quad + \sum_{j=1}^\ell\sum_{k=1}^j (H+M)^{-k} V (H_0+M)^{-j-1+k} V
(H_0+M)^{-\ell-1+j} \\
&= \cdots.
\end{align*}
Iterating this procedure $\ell$-times, and using the fact\footnote{\ See Footnote~\ref{traceideal}.} 
$V(H_0+M)^{-j}\in\mathcal {I}_p$
for $p>n/(2j)$, we learn that $A-A_0\in\mathcal {I}_1$.
Then the main part of the proof of
Proposition~\ref{basicpro} can be modified accordingly.

In order to modify the proof of Lemma~\ref{lem:2RB}, we use
\begin{align*}
A_0-\m(z) &= -\sum_{j=1}^\ell (z+M)^{-j} (H_0-z) (H_0+M)^{-\ell-1+j} \\
&= -(H_0-z)(H_0+M)^{-1}L_0(z),
\end{align*}
where $\m(z)=(z+M)^{-\ell}$, and we have written  
\[
L_0(z)=\sum_{j=1}^\ell (z+M)^{-j} (H_0+M)^{-\ell+j}. 
\]
Since $\Re(z+M)>M$ if $z\sim\l>0$, $L_0(z)$ is invertible.
In consequence, we obtain
\begin{align*}
(A_0-\m(z))^{-1} &= -L_0^{-1}(z)(H_0+M)(H_0-z)^{-1} \\
&= -L_0^{-1}(z)\left[1+(z+M)(H_0-z)^{-1}\right].
\end{align*}
We also write 
\[
L(z)=\sum_{j=1}^\ell (z+M)^{-j} (H+M)^{-\ell+j}.
\]
Then we have
\begin{align*}
&(A_0-\m(z))^{-1}(A-\m(z))^{-1}\\  
&=L_0^{-1}(z)\Bigl\{a_0(z) +a_1(z)(H_0-z)^{-1} 
+a_2(z) (H_0-z)^{-2}\\ 
&+a_3(z) (H_0-z)^{-1} V(H-z)^{-1}+
a_4(z) (H_0-z)^{-2} V(H-z)^{-1}\Bigr\} L^{-1}(z)
\end{align*}
with some polynomials $a_j(z)$ in $z$. 
Moreover, using the standard weight estimates,
\[
\jap{x}^\c (H_0+M)^{-1}\jap{x}^{-\c}, \quad
\jap{x}^\c (H+M)^{-1}\jap{x}^{-\c} \in B(L^2(\re^n)),
\]
we can carry out the same argument as in the proof of 
Lemma~\ref{lem:2RB}. Consequently, we have: 

\begin{prop}
Let $\l\in(0,\infty)$. Then
\[
\lim_{z\to \l+i0} \trace\left[(A-\m(z))^{-1} - (A_0-\m(z))^{-1}\right]
\]
exists, and the limit is continuous in $\l$ in $(0,\infty)$.
Moreover, the SSF $\x(\l)$ exists for $\l\in
(0,\infty) $, and $\x(\cdot)$
is continuous in $(0,\infty) $.
\end{prop}

\section{The Friedel Sum Rule}
\label{sec:FSR}

In solid state physics \cite{Kittel}, the difference of the number of the states 
given by the right-hand side of (\ref{xiformal}) has been often called the excess charge. 
In this section, we define the excess charge, and show 
that it is equivalent to the SSF. Besides, the SSF is equal to the total phase shift $\th(\l)$ 
which is given by 
\[
e^{2\pi i\th(\l)} = \det S(\l), \quad \l>0,
\]
where $S(\l)$ is the scattering matrix.
By the invariance principle and the Birman-Krein formula \cite{BY}, we have 
\[
\th(\l)=\x(\l;H,H_0)=-\x(\m(\l);A,A_0) 
\]
with $\th(\l)=0$ for $\l<\sigma(H)$. 
Therefore, the excess charge is equal to the total phase shift. 
This is nothing but the Friedel sum rule. 

To begin with, we introduce a cutoff function $\vartheta_R(x) =\vartheta(x/R)$ 
with a large $R>0$ and with $\vartheta\in\ C_0^\infty(\re^n)$ 
satisfying $\vartheta=1$ in a neighborhood of $x=0$. 
Then the excess charge is defined by
\[
Z(\lambda):=\lim_{R\to\infty} \trace\>\left[\vartheta_R(E_H(\l)-E_{H_0}(\l))
\vartheta_R\right],
\]
where $E_A(\l)$ denotes the spectral projection: $\i_{(-\infty,\l]}(A)$.
We want to show that the above limit exists, and that it is
equivalent to the SSF under certain assumptions.

We denote 
\[
Z_R(\l)=\trace\> \left[\vartheta_R (E_H(\l)-E_{H_0}(\l))\vartheta_R\right]
\]
for $\l>0$. Using the notation of Section~\ref{sec:SSF}, we recall that 
\[
E_H(\l)=1-E_A(\m(\l)) =-\lim_{z\to\m(\l)+i0}\Im \frac{1}{\pi}\int_{\c_z}
(A-w)^{-1}dw
\]
in the strong sense. We have 
\begin{align*}
&\vartheta_R (E_H(\l)-E_{H_0}(\l))\vartheta_R \\
&\qquad = -\lim_{z\to\m(\l)+i0} \Im \frac{1}{\pi} \int_{\c_z}
\vartheta_R\left[(A-w)^{-1}-(A_0-w)^{-1}\right]\vartheta_Rdw\\
&\qquad =\lim_{z\to\m(\l)+i0} \Im \frac{1}{\pi} \int_{\c_z}
\bigbrac{\vartheta_R(A-w)^{-1}\jap{x}^{-\b}W
\jap{x}^{-\b}(A_0-w)^{-1}\vartheta_R}dw
\end{align*}
in the same way as in the proof of Proposition~\ref{basicpro}. 
The integrand $[\cdots]$ in the right-hand side of the second equality is 
of the trace class, and is continuous in $w$ up to the boundary. 
Actually, $W\in \mathcal{I}_1$ as we proved in the preceding section, 
and we also have, in the same way, 
that $\vartheta_R(A-w)^{-1}\jap{x}^{-\b}$ and $\jap{x}^{-\b}(A_0-w)^{-1}\vartheta_R$ 
have norm limits as $w\to\m(\l)+i0$. Thus we learn that 
\[
Z_R(\l) =\Im \frac{1}{\pi} \int_{\c_{\m(\l)}}
\trace\,\bigbrac{\vartheta_R(A-w)^{-1}\jap{x}^{-\b}W
\jap{x}^{-\b}(A_0-w)^{-1}\vartheta_R} dw.
\]
In particular, this implies the existence of $Z_R(\l)$.
We show

\begin{thm}
\label{thm:EXCH}
Let $\l\in(0,\infty)$. Then
\[
\lim_{R\to\infty} Z_R(\l) =\x(\l;H,H_0).
\]
\end{thm}

\begin{proof}
{From} (\ref{expSSF}), (\ref{diffResolA}) and the above representation of $Z_R(\l)$, we
have
\begin{align*}
&Z_R(\l)-\x(\l;H,H_0) \\
&=Z_R(\l)+\x(\mu(\l);A,A_0) \\
&=\Im \frac{1}{\pi}\int_{\c_{\m(\l)}} \trace\bigbrac{ W\jap{x}^{-\b}(A_0-w)^{-1}
(\vartheta_R^{\ 2}-1)(A-w)^{-1}\jap{x}^{-\b}}dw.
\end{align*}
Therefore, it suffices to show
\[
\bignorm{\jap{x}^{-\b}(A_0-w)^{-1} (1-\vartheta_R^{\ 2})(A-w)^{-1}\jap{x}^{-\b}}
\to 0\quad \text{as $R\to\infty$}
\]
uniformly in $w\in{\c_{\m(\l)}}$. If $w$ is away from $\s(A)\cup\s(A_0)$, then
\begin{align*}
&\bignorm{\jap{x}^{-\b}(A_0-w)^{-1} (1-\vartheta_R^{\ 2})(A-w)^{-1}\jap{x}^{-\b}} \\
&= \bignorm{\jap{x}^{-\b} (A_0-w)^{-1}\jap{x}^{\b} (\jap{x}^{-\b}
(1-\vartheta_R^{\ 2}))
(A-w)^{-1}\jap{x}^{-\b} } \\
&\leq\bignorm{\jap{x}^{-\b} (A_0-w)^{-1}\jap{x}^{\b} }\cdot
\norm{\jap{x}^{-\b} (1-\vartheta_R^{\ 2})}\cdot
\bignorm{(A-w)^{-1} }=O(R^{-\b})
\end{align*}
locally uniformly in $w$. Thus it suffices to consider the case
$w\sim\mu(\l)\pm i0$.

As well as in \S\S2.2, we have 
\begin{align*}
&\jap{x}^{-\b} (A_0-\m(z))^{-1}(1-\vartheta_R^2)(A-\m(z))^{-1}\jap{x}^{-\b} \\
&\qquad = \bigpare{\jap{x}^{-\b} L_0(z)^{-1} \jap{x}^\b}\times 
\jap{x}^{-\b} \bigbrac{1+(z+M)(H_0-z)^{-1}} (1-\vartheta_R^2)\times \\
&\qquad\qquad  \times \bigbrac{1+(z+M)(H-z)^{-1}} \jap{x}^{-\b}
\times \bigpare{\jap{x}^{\b} L(z)^{-1}\jap{x}^{-\b}}. 
\end{align*}
Therefore, it is enough to show 
\[
\bignorm{\jap{x}^{-\b}(H_0-z)^{-1} (1-\vartheta_R^{\ 2})(H-z)^{-1}\jap{x}^{-\b}}
\to 0\quad \text{as $R\to\infty$}
\]
if $z\sim \l\pm i0$ in $\co_\pm =\{ z\,|\, \pm\Im z\geq 0\}$. Since
\begin{align*}
&\jap{x}^{-\b}(H_0-z)^{-1}(1-\vartheta_R^{\ 2})(H-z)^{-1}\jap{x}^{-\b} \\
&=\jap{x}^{-\b} (H_0-z)^{-1}(1-\vartheta_R^{\ 2})(H_0-z)^{-1} \jap{x}^{-\b} \\
&\quad
-\jap{x}^{-\b}(H_0-z)^{-1}(1-\vartheta_R^{\ 2})(H_0-z)^{-1}V(H-z)^{-1}\jap{x}^{-\b},
\end{align*}
Theorem~\ref{thm:EXCH} now follows from the next lemma.
\end{proof}

\begin{lem}
There exist $\mathcal {U}$: a neighborhood of $\l\pm i0$ in $\co_\pm$,
$\e>0$ and
$C>0$ such that
\[
\bignorm{\jap{x}^{-\b} (H_0-z)^{-1}(1-\vartheta_R^{\ 2}) (H_0-z)^{-1} \jap{x}^{-\b}
} \leq C R^{-\e}
\]
for $\beta>3/2$ and $z\in\mathcal {U}$.
\end{lem}

\begin{proof}
We consider the case $\l+i0$ only. The other case is similar.
It is easy to observe that it suffices to show
\[
\bignorm{\jap{x}^{-\b} (H_0-z)^{-1}(1-\vartheta_R^{\ 2}) \y(H_0)(H_0-z)^{-1}
\jap{x}^{-\b} } \leq C R^{-\e},
\]
where $\y\in C_0^\infty((0,\infty))$ such that $\y=1$ in a neighborhood
of $\l$.
In order to show this, we use a Mourre-type microlocal resolvent estimate of
Isozaki-Kitada \cite{IK}. See also \cite{M}. We apply their result with $H=H_0$. 
Let $\rho_\pm\in C^\infty([-1,1])$ such that
\[
\rho_+(t)+\rho_-(t)=1; \quad
\rho_\pm(t)= 0 \quad\text{if } \pm t<-\frac12.
\]
We also set $\d$ and $\e$ so that
\[
\frac12 <\d<\b-1; \quad 0<\e<\d-\frac{1}{2}.
\]
We write 
\[
p_\pm(x,\x) = R^\e \jap{x}^{-\e}\left\{1-[\vartheta_R(x)]^2\right\}\rho_\pm(\hat x\cdot
\hat\x)\y(|\x|^2),
\]
where $\hat x=x/|x|$. We quantize $p_\pm$ by the usual Kohn-Nirenberg
pseudodifferential
operator calculus:
\[
P_\pm f(x) =p_\pm(x,D_x)f(x)
= (2\pi)^{-n/2} \int p_\pm(x,\x) e^{ix\cdot\x} \hat f(\x)d\x.
\]
Then we have
\[
(1-\vartheta_R^{\ 2})\y(H_0) = R^{-\e}\jap{x}^\e (P_++P_-),
\]
and $p_-(x,\x)$ satisfies the assumptions of Theorem~1 (or Theorem~1.2)
of \cite{IK},
uniformly in $R>1$, and we obtain
\[
\bignorm{\jap{x}^\d P_- (H_0-z)^{-1} \jap{x}^{-\b} } \leq C
\quad\text{for }z\sim\l, z\in\co_+.
\]
Similarly, we can apply the same argument to $(\jap{x}^\e P_+\jap{x}^{-\e})^*$ 
instead of $P_-$. In fact, $(\jap{x}^\e P_+\jap{x}^{-\e})^*$ is also a pseudodifferential 
operator, and its symbol can be computed by asymptotic expansions, up to an error of 
$O(\jap{x}^{-\infty})$. Let $\tilde p_+(x,\x)$ be the symbol of  $(\jap{x}^\e P_+\jap{x}^{-\e})^*$.
Then it has the same support property with $p_+(x,\x)$. In particular, we can show 
\[
\tilde p_+(x,\x)= p_+(x,\x)+ i(\pa_x\cdot \pa_\x) p_+(x,\x) -i \pa_\x p_+(x,\x)\cdot (\e x/\jap{x}^2) +O(\jap{x}^{-2}).
\]
Thus we obtain
\[
\bignorm{\jap{x}^{-\b} (H_0-z)^{-1} \jap{x}^\e P_+\jap{x}^{\d-\e}} \leq C
\quad\text{for }z\sim\l, z\in\co_+, 
\]
as well. Combining these, we have
\begin{align*}
&\bignorm{\jap{x}^{-\b} (H_0-z)^{-1}(1-\vartheta_R^{\ 2})\y(H_0)(H_0-z)^{-1}
\jap{x}^{-\b} }\\
&\quad= R^{-\e} \bignorm{\jap{x}^{-\b} (H_0-z)^{-1}\jap{x}^\e
(P_++P_-)(H_0-z)^{-1}\jap{x}^{-\b}}\\
&\quad\leq R^{-\e} \Bigbra{\bignorm{\jap{x}^{-\b}
(H_0-z)^{-1}\jap{x}^{-(\d-\e)}}\cdot
\bignorm{\jap{x}^\d P_- (H_0-z)^{-1} \jap{x}^{-\b}}\\
&\qquad\quad +\bignorm{\jap{x}^{-\b}(H_0-z)^{-1} \jap{x}^\e
P_+\jap{x}^{\d-\e}}\cdot
\bignorm{\jap{x}^{-(\d-\e)}(H_0-z)^{-1}\jap{x}^{-\b}}} \\
&\quad\leq CR^{-\e}
\end{align*}
for $z\sim\l$, $z\in\co_+$.
\end{proof}

\thanks{\textbf{Acknowledgement:} It is a pleasure to thank the following people for 
discussions and correspondence: Jacques Friedel, Arne Jensen and Hal Tasaki. SN is partially supported by 
JSPS Grant Kiban (A) 21244008 (2009-2013). }


\end{document}